\documentclass[letterpaper, 10 pt, conference]{ieeeconf}
\IEEEoverridecommandlockouts
\usepackage{graphicx} 
\usepackage{cite}
\usepackage{amssymb}
\usepackage{amsmath}
\usepackage{enumerate}
\usepackage{accents}
\usepackage[dvipsnames]{xcolor}

\newtheorem{theorem}{Theorem}[section]

\newtheorem{lemma}[theorem]{Lemma}

\newtheorem{remark}[theorem]{Remark}
\newtheorem{proposition}[theorem]{Proposition}

\newcommand{\R}{\mathbb{R}}
\newcommand{\trace}{\text{tr}}

\title{\LARGE \bf
Tunable Thresholds and Frequency Encoding\\ in a Spiking NOD Controller
}

\author{Ian Xul Belaustegui$^{a}$, 
        Alessio Franci$^{b}$, 
        Naomi Ehrich Leonard$^{a}$
    \thanks{}
    \thanks{This research has been supported in part by ONR grant N00014-19-1-2556 and the William R.\ and Jane G.\ Schowalter Research Fund.}
    \thanks{$^{a}$ Department of Mechanical and Aerospace Engineering, Princeton University; \texttt{\{ianxul,naomi\}@princeton.edu}}%
    \thanks{$^{c}$ Department of Electrical Engineering and Computer Science, University of Liege, and WEL Research Institute, Wavre, Belgium; \texttt{afranci@uliege.be}}%
}

\begin{document}

\maketitle

\begin{abstract}
     Spiking Nonlinear Opinion Dynamics (S-NOD) is an excitable decision-making model inspired by the spiking dynamics of neurons. S-NOD enables the design of agile decision-making that can rapidly switch between decision options in response to a changing environment. In S-NOD, decisions are represented by discrete opinion spikes that evolve in continuous time. Here, we extend previous analysis of S-NOD and explore its potential as a nonlinear controller with a tunable balance between robustness and responsiveness to input. We identify and provide necessary conditions for the bifurcation that determines the onset of periodic opinion spiking. We leverage this analysis to characterize the tunability of the input-output threshold for opinion spiking as a function of the model basal sensitivity and the tunable dependence of opinion spiking frequency on input magnitude above the threshold. We conclude with a discussion of S-NOD as a new neuromorphic control block and its extension to distributed spiking controllers.  
\end{abstract}

\section{Introduction}
Decision-making—the integration of sensory information and internal objectives to select an action—is a fundamental process in biological systems, operating across all levels of organization, from molecular networks~\cite{Balázsi_CellDecisions_2011} to neural populations~\cite{Hein_EcoDecision_2022} and collective behaviors~\cite{Sridhar_____Couzin_2021}. Modeling biological decision-making is crucial not only for understanding its underlying mechanisms but also for engineering bio-inspired autonomous systems.

Traditional models of opinion dynamics often struggle to capture the robust yet flexible decision-making behavior observed in biological systems. The Nonlinear Opinion Dynamics (NOD) model \cite{Bizyaeva_NOD_2023, Leonard_NOD_ARpaper2024} shows how balancing negative and positive feedback can lead to robust, fast, and flexible decision-making through an organizing pitchfork bifurcation.
However, NOD may not always be able to switch \textit{between} decision options with sufficient agility in rapidly changing environments. The Spiking-NOD (S-NOD) model~\cite{Cathcart_SNOD_2024} addresses this limitation by making NOD excitable, in the sense that it has a roughly linear response to input below a threshold, but can produce large excursions in phase-space (\textit{spikes}) for input above the threshold~\cite{Izhikevich_DynamicalNeuro_2007}. Excitability is achieved through adaptive slow negative feedback that brings the NOD state back to its neutral, ultra-sensitive state each time a decision is made, resulting in excitable {\it opinion spikes}. 

As an excitable system, S-NOD exhibits a form of controlled instability, where positive feedback leads to large and fast deviations away from equilibrium while slow negative feedback ensures regulation toward it. This mechanism has been successfully applied to robot navigation in crowded environments, ensuring responsive and decisive behavior~\cite{Cathcart_SNOD_2024}. The use of excitability for control is central to neuromorphic engineering~\cite{Sepulchre_Spiking_2022}, with applications demonstrated in tasks such as the rhythmic control of a pendulum~\cite{Schmetterling_NeuroPendulum_2024}. Fig.~\ref{fig:CharlotteFig1}, adapted from \cite{Cathcart_SNOD_2024}, illustrates the agility of the S-NOD controller.

Here, we extend the analysis in~\cite{Cathcart_SNOD_2024} to rigorously characterize two key tunable features of S-NOD useful for control: 
the spiking threshold and the response of spiking frequency to input magnitude. 
Section \ref{sec:S-NOD} introduces S-NOD and the notation for our analysis. In Section \ref{sec:Hopf-bif}, we study the dependence of S-NOD dynamics on exogenous input and prove that spiking emerges through a Hopf bifurcation at a tunable input threshold. In Section \ref{sec:PF-bif}, we study deadlock breaking where two symmetric spiking limit cycles appear in the absence of input as the system undergoes a symmetry-breaking pitchfork bifurcation at the critical basal sensitivity. In Section \ref{sec:tunability&freq}, we use geometric arguments to characterize how the spiking threshold and frequency response to input can be tuned as a function of model parameters. This paper thus provides an in-depth dynamical systems characterization of S-NOD as a tunable and agile spiking control element.

\begin{figure}
    \centering
    \includegraphics[width=1.0\linewidth]{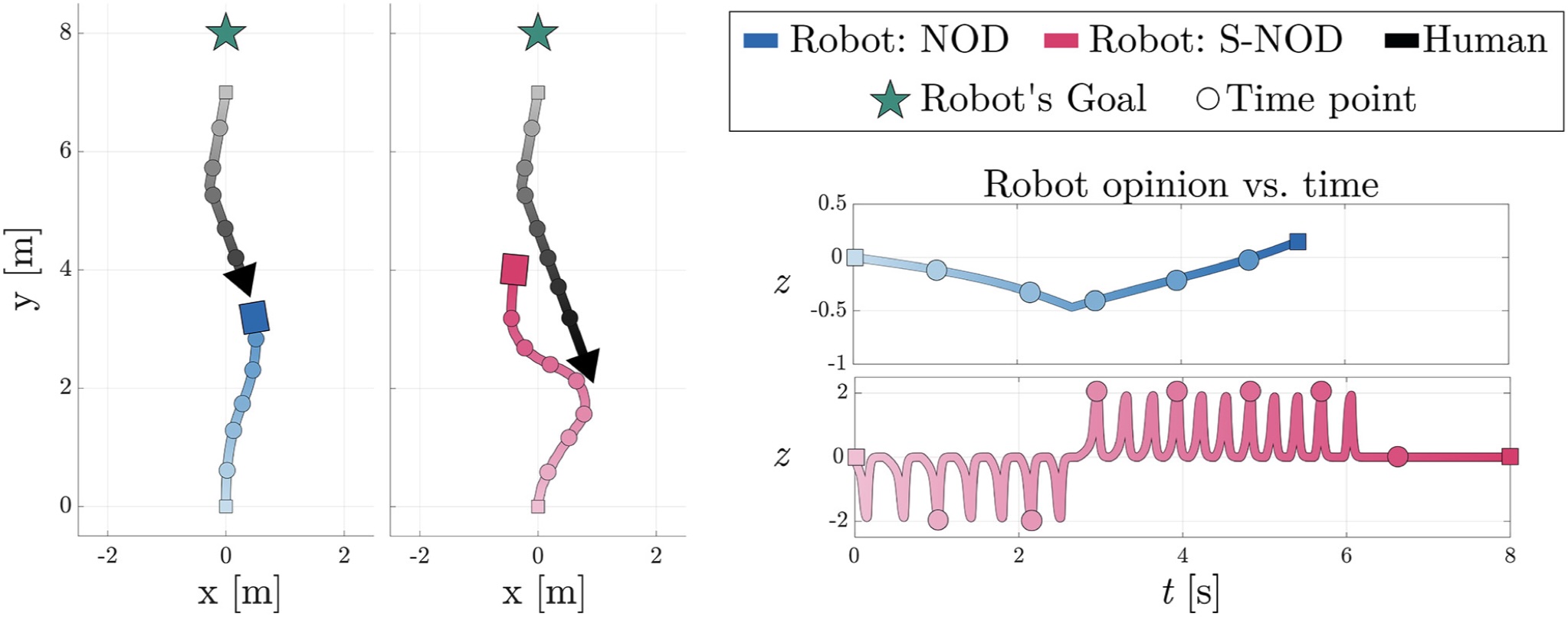}
    \caption{Figure from \cite{Cathcart_SNOD_2024} shows S-NOD (pink) applied to control a planar robot in a social navigation scenario. NOD (blue) is shown for comparison.  
    The robot aims to reach a goal marked by a star while avoiding collision with an unpredictable oncoming mover (black). The NOD controller is slow to respond to an abrupt change in the direction of the oncoming mover, but the S-NOD controller responds with agility and avoids colliding. }
    \label{fig:CharlotteFig1}
\end{figure}


\section{Spiking Nonlinear Opinion Dynamics}\label{sec:S-NOD}

S-NOD is governed by the  differential equations
\begin{align}\label{eq:sys}
    \dot z &= -d\,z + \tanh((kz^2+\mu_0-s)a z+b)  \\
    \dot s &= \varepsilon(-s +k_sz^4)\,. \notag
\end{align}
The state $z(t)\in\R$ is the agent's opinion at time $t$. For decision-making between two mutually exclusive options, we interpret $z(t)>0$ ($z(t)<0$) as an opinion in favor of option 1 (option 2). $|z(t)|$ is the strength of the opinion and $z(t)=0$ is a neutral opinion (indecision). $s(t)\in\R$ is the agent's slow recovery state at time $t$. $d \in\R_+$ is the gain of negative feedback opinion-damping term $-d\,z$. $\mu_0\in\R_+$ sets the basal level of the agent's sensitivity\footnote{In \cite{Cathcart_SNOD_2024} $\mu$, $\mu_0$ are  $u$, $u_0$ and called ``attention'' and ``basal attention.''} $\mu$ to positive feedback term $az$, where $\mu = kz^2+\mu_0-s$. 
The agent's sensitivity $\mu$ increases with gain $k\in\R_+$ as the magnitude of the agent's opinion increases. At a much slower time-scale, $\mu$  decreases as the slow recovery variable $s$ increases. This creates a balance of fast positive and slow negative feedback. $\varepsilon\in\R_+$, $\varepsilon\ll1$ sets the timescale separation between the dynamics of $z$ and $s$: $z$ is fast and $s$ is slow.  $\mu_o,a,k\in\R_+$ tune the positive feedback gain. $k_s\in\R_+$ tunes the slow negative feedback gain. $b\in\R$ is an exogenous input. In applications $b$ would integrate various task-relevant signals, and represent evidence for a decision in the positive or negative $z$ direction. In the following analysis we consider $b$ to be a constant parameter. For our results to apply directly in situations where $b$ is a time-varying signal, we would require it to have a timescale slower than the dynamics  (\ref{eq:sys}). 

We first establish some basic results needed to analyze the onset of opinion spiking as either input $b$ (in Section~\ref{sec:Hopf-bif}) or tuning parameter $\mu_0$ increases (in Section~\ref{sec:PF-bif}). 

\begin{remark}
    If the point $(\bar z, \bar s) \in \R^2$ is an equilibrium of (\ref{eq:sys}) then $\bar s = k_sz^4$, and $\bar z$ solves the equation
    \begin{equation}\label{eq:imp-sol}
        h_b(\bar z):=-d\bar z + \tanh(a \bar z (-k_s\bar z^4+k\bar z^2 +\mu_0)+b)=0.
    \end{equation}
    We will 
say $\bar z\in\R$ is an equilibrium if it satisfies \eqref{eq:imp-sol}.
\end{remark}

The stability of equilibria of two-dimensional systems is fully determined by the trace and determinant of the Jacobian evaluated at equilibrium. Let $\varphi(z,s):=a z (k z^2 + \mu_0 - s)+b$ and $\psi_b(z):=\varphi(z, k_s z^4) = a z(-k_s z^4 + k z^2 + \mu_0 )+b$.
Letting $\tanh'(x)=\frac{d\tanh(x)}{dx}=1-\tanh^2(x)$, the Jacobian of~\eqref{eq:sys} is
\begin{align}
    &J(z,s)\!:=\!\begin{pmatrix}
        -d +\tanh'(\varphi)\frac{\partial \varphi}{\partial z} & \tanh'(\varphi)\frac{\partial \varphi}{\partial s} \\
        4\varepsilon k_s z^3 & -\varepsilon
    \end{pmatrix}\,.
\end{align}
\begin{remark}\label{rem:tanh-der}
    If $\bar z\in\R$ is an equilibrium then $\tanh'(\varphi(\bar z, k_s\bar z^4))=1-d^2\bar z^2$.
\end{remark}

For notation purposes let $J_z:=J(z, k_sz^4)$ be the Jacobian evaluated over the $s$-nullcline. The following result follows from elementary algebraic manipulations. 

\begin{lemma}
    If $\bar z$ is an equilibrium, then
    \begin{align}
        \trace (J_{\bar z}) &= \underbrace{ad^2k_s}_{c_3}\bar z^6-\underbrace{(3ad^2k+ak_s)}_{c_2}\bar z^4 \notag\\
        &+\underbrace{(3ak-ad^2\mu_0)}_{c_1}\bar z^2+\underbrace{a\mu_0-d}_{c_0} - \varepsilon \label{eq:traceJ}\\
         \frac{\det (J_{\bar z})}{\varepsilon} &= -\underbrace{5ad^2k_s}_{\hat c_3=5c_3}\bar z^6+\underbrace{(3ad^2k+5ak_s)}_{\hat c_2=c_2+4ak_s}\bar z^4\notag\\&-\underbrace{(3ak-ad^2\mu_0)}_{\hat c_1=c_1}\bar z^2-\underbrace{(a\mu_0-d)}_{\hat c_0=c_0}. \label{eq:detJ}
\end{align}
\end{lemma}
Finding the roots of~\eqref{eq:traceJ} and \eqref{eq:detJ} reduces to solving two third-degree polynomials, namely,  
$q_{tr}(z):=\trace (J_{\sqrt{z}})$ and $q_{det}(z):=\det (J_{\sqrt{z}})/\varepsilon$. 


\section{Spiking from increasing $b$}\label{sec:Hopf-bif}

\begin{figure}
    \centering
    \includegraphics[width=1.0\linewidth]{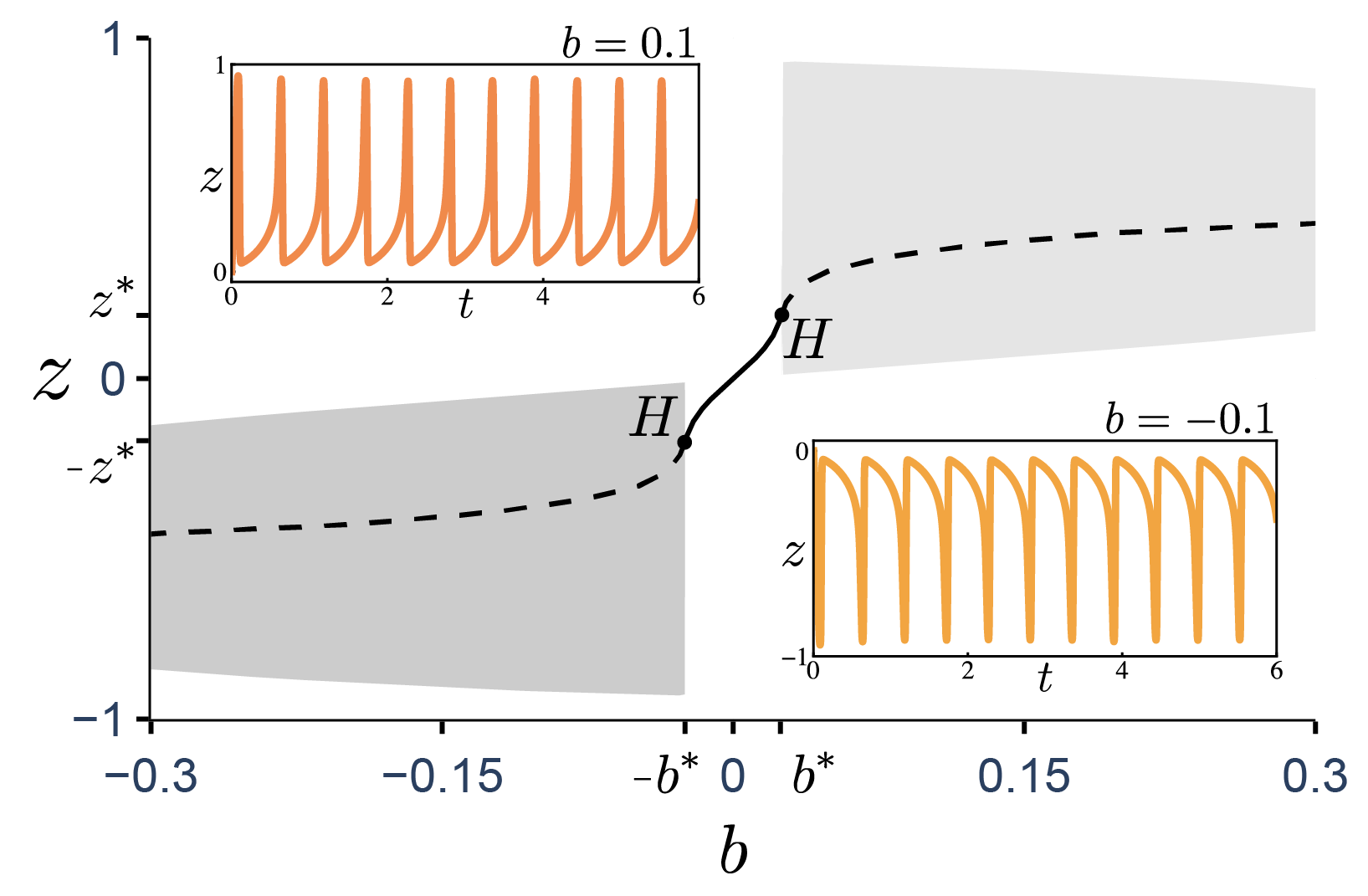}
    \caption{Bifurcation diagram for (\ref{eq:sys}) with respect to input ($b$) projected onto the opinion ($z$) axis to reduce from three dimensions to two. Solid lines indicate stable equilibria, dashed lines indicate unstable equilibria, shaded areas indicate stable spiking limit cycles, vertical width of shaded area represents spike amplitude. Black circles marked with $H$ correspond to the Hopf bifurcations at $b=\pm b^*$. Inset plots show spiking dynamics for constant suprathreshold positive ($b=0.1$) and negative ($b=-0.1$) input. Parameters used are $\mu_0 = 0.8, a = 1, d = 1, k = 3, k_{s} = 22, \varepsilon = 0.05$.}
    \label{fig:b-bifurcation}
\end{figure}

We characterize S-NOD dynamics as a function of the input value $b$. The main result, Thm. \ref{thm:hopf}, shows that S-NOD transitions, through a Hopf bifurcation, from rest to spike as the input magnitude $|b|$ increases above a critical value $b^*>0$, determined by the model parameters. 
Thm.~\ref{thm:hopf} characterizes S-NOD, the spiking threshold and its sensitivity to model parameters. Additional results presented in this section are preliminary to Thm.~\ref{thm:hopf} and further characterize S-NOD dynamics. Fig.~\ref{fig:b-bifurcation} shows a numerical bifurcation diagram, predicted by Thm.~\ref{thm:hopf}, revealing spiking limit cycles for increasing input magnitude in both the positive (upward spikes) and negative (downward spikes) input directions.
All analysis in this section is presented for constant $b\geq 0$. The case $b\leq 0$ follows directly from the model's odd symmetry.



\begin{remark}\label{obs:im-eigs-char}
    A Hopf bifurcation happens at an equilibrium having a pair of conjugate non-zero purely imaginary eigenvalues. At an equilibrium $(\bar z, k_s\bar z^4)$ of (\ref{eq:sys}) this happens if and only if $\trace (J_{\bar z})=0$ and $\det(J_{\bar z})>0$.
\end{remark}


Although explicit expressions for equilibria of (\ref{eq:sys}) are not computable, we can provide conditions under which (\ref{eq:sys}) has a unique equilibrium.
\begin{proposition}\label{prop:single-sol}
    Assume that $\det(J_z)>0$ for all $\lvert z\rvert<d^{-1}$. Then, for all $b\in\R$, (\ref{eq:sys}) has a unique equilibrium.
\end{proposition}
\begin{proof}
    Let $b\geq 0$. We first show that there is at least one equilibrium. If $b=0$ then $h_0(0)=0$. Otherwise, if $b>0$ then  $h_b(0)=\tanh(b)>0$, and $\lim_{z\to\infty}h_b(z)=-\infty$, which by continuity implies the existence of $\bar z>0$ such that $h_b(\bar z)=0$. 
    Observe that, since $|\tanh(x)|<1$, 
    any equilibrium $\bar z$ of~(\ref{eq:sys}) satisfies $\lvert\bar z\rvert <d^{-1}$ and that, invoking the hypothesis $\det(J_z)>0$,
\begin{align}
    h_b'(\bar z)=-d+(1-d^2\bar z^2)\psi_b'(\bar z)  =-{\det (J_{\bar z}})/{\varepsilon}<0.\label{eq:h'_and_derJ}
\end{align}


This implies all roots of $h_b$ have strictly negative derivative, which in turn implies there cannot be more than one root (having two consecutive roots with derivatives of the same non-zero sign would lead to a contradiction).
\end{proof}

\begin{remark}
    The assumption $\det(J_z)>0$ for all $\lvert z\rvert<d^{-1}$ in Proposition~\ref{prop:single-sol} can be verified by studying positive definiteness of the third degree polynomial $q_{det}$ on the interval $[0,d^{-2})$ and is always satisfied for sufficiently small values of $\mu_0$ and $k$.
\end{remark}


The next proposition shows that $\bar z$ increases monotonically with $b$, and asymptotically to $d^{-1}$ as $b\to\infty$.

\begin{proposition}\label{prop:sol-lim}
    Assume that $\det(J_z)>0$ for all $\lvert z\rvert<d^{-1}$. Given $b\in\R$, let $\bar z_b\in\R$ be the unique zero of $h_b$. Then $\lim_{b\to\infty}\bar z_b=d^{-1}$ and $\bar z_b$ is strictly monotonically increasing with $b$. 
\end{proposition}
\begin{proof}
    Observe that $\forall z\in\mathbb{R},\lim_{b\to\infty}h_b(z)=-d\,z+1$. 
    Hence, $\lim_{b\to \infty}\bar z_b=d^{-1}$.
    Using~\eqref{eq:h'_and_derJ}, by implicit differentiation it follows that
    \begin{align*}
           \frac{d \bar z_b}{d b}&=-\frac{\partial h_b(\bar z_b)/\partial b}{\partial h_b(\bar z_b)/\partial \bar z_b} = \frac{\varepsilon(1-d^2\bar z_b^2)}{\det (J_{\bar z_b})},
    \end{align*} 
    where the last equality follows by Remark \ref{rem:tanh-der} and  \eqref{eq:h'_and_derJ}. Since $\lvert\bar z_b\rvert<d^{-1}$ and $\det(J_{\bar z_b})>0$, then $\forall b\in\R_+, \frac{\partial \bar z_b}{\partial b}>0$. 
\end{proof}


Under the same assumptions of Props.~\ref{prop:single-sol} and~\ref{prop:sol-lim}, in order to characterize changes of stability of the unique equilibrium of~(\ref{eq:sys}) it is sufficient to study if and where $\trace(J_z)$ changes sign on the interval $(-d^{-1}, d^{-1})$.
The following proposition provides tight sufficient conditions for this. Let $\xi_\pm:=\frac{c_2\pm\sqrt{c_2^2-3c_1c_3}}{3c_3}$ be the two roots of $q_{tr}'$ (i.e., the two inflection points of $q_{tr}$).

\begin{proposition}\label{prop:trace-roots}
    If $\mu_0 <\frac{d}{a}$, $k> \frac{d^3}{3a}$, and $\trace (J_{\sqrt{\xi_-}})>0$, then there exist $z^{\ast}, z^{\dagger}\in (0,d^{-1})$ such that $\trace (J_{z^{\ast}})=\trace (J_{z^{\dagger}})=0$ and $\forall z\in[0, z^{\ast})\cup(z^{\dagger}, d^{-1}], \trace (J_{z})<0$, and $\forall z\in(z^{\ast}, z^{\dagger}), \trace (J_{z})>0$. 
\end{proposition}
\begin{proof}
    For all positive parameters, $c_2^2-3c_1c_3\geq 0$, and so $\xi_\pm$ are real. The condition $k> \frac{d^3}{3a}>\frac{d^2\mu_0}{3}$ implies $c_1> 0$, so $\xi_->0$. Since $c_3>0$ is the leading coefficient of $q_{tr}$ (see (\ref{eq:traceJ})),
    we have $\lim_{z\to\infty}q_{tr}(z)=\infty$. Also $q_{tr}(d^{-2})=-d-\varepsilon<0$, and the assumption $\frac{d}{a}>\mu_0$ implies $q_{tr}(0)=c_0-\varepsilon<0$. So $\xi_-<\xi_+$ are the local minimum and maximum respectively of this third degree polynomial. Thus, in the interval $(0,d^{-2})$, $q_{tr}$ goes from negative, to positive, to negative. Since $\trace(J_{z})=q_{tr}(z^2)$, we conclude there exist two roots $z^{\ast}, z^{\dagger}\in(0,d^{-1})$, $z^{\ast} < z^{\dagger}$, which satisfy the properties for $\trace(J_z)$ in the proposition statement. 
\end{proof}


\begin{remark}
\label{rem:range_of_mu}
    Condition $\mu_0 <\frac{d}{a}$ ensures that without input the neutral opinion state is stable. Condition $k> \frac{d^3}{3a}$ ensures the $z$-nullcline has the shape of a symmetric subcritical pitchfork bifurcation, as shown in \cite{Cathcart_SNOD_2024}. Condition $\trace (J_{\sqrt{\xi_-}})>0$ ensures there are two intervals of possible equilibria with positive values of the trace. Since $\trace (J_{\sqrt{\xi_-}})$ is strictly monotonically increasing (decreasing) with $\mu_0$ ($\varepsilon$), for sufficiently small $\mu_0$ (large $\varepsilon$) the unique equilibrium is always stable. 
    However, since $\det(J_z)$ is strictly monotonically decreasing with $\mu_0$, the assumptions of Props.~\ref{prop:single-sol} and~\ref{prop:sol-lim} that $\det(J_z)>0$ can be violated for too large $\mu_0$ (even if less than $\frac{d}{a}$), and thus our guarantees on uniqueness of the neutral equilibrium would no longer hold. In the case that $\mu_0<\frac{d}{a}$ but $\det(J_z)<0$ for some points, the equilibrium that starts at the origin when there is no input will still grow monotonically with $b$, but eventually disappears in a Saddle-Node bifurcation .
\end{remark}

The following theorem characterizes bifurcations of~(\ref{eq:sys}) with respect to $b$ for intermediate values of $\mu_0$ and small $\varepsilon$.

\begin{theorem}[Spiking Input Threshold]\label{thm:hopf}\ \\
    Assume $\mu_0 < \frac{d}{a}$,  $k> \frac{d^3}{3a}$, and $\epsilon >0$ small such that $\trace (J_{\sqrt{\xi_-}})>0$ and $\det(J_z)>0$ for all $\lvert z\rvert < d^{-1}$. Then there exist exactly two critical values of input $b^*,b^{\dagger}\in\R_+$, with $b^*< b^{\dagger}$, at which the unique equilibrium $\bar z_b$ of~(\ref{eq:sys}) undergoes a Hopf bifurcation, i.e., at $b=b^{\ast}$ and $b=b^{\dagger}$ (and symmetrically at $b=-b^{\ast}$ and $b = -b^{\dagger}$). 
\end{theorem}
\begin{proof}
    By Prop.~\ref{prop:trace-roots} there exist $z^{\ast},z^{\dagger}\in(0,d^{-1})$ roots of $\trace (J_z)$ as in the statement, with $z^{\ast}<z^{\dagger}$. We have $\bar z_0 = 0$, and by Prop.~\ref{prop:sol-lim} we know $\lim_{b\to\infty}\bar z_b=d^{-1}$. So by continuity there exist $b^*, b^{\dagger}\in\R_+$, with $b^*< b^{\dagger}$, such that $\bar z_{b^*}=z^{\ast}$ and $\bar z_{b^{\dagger}}=z^{\dagger}$. Thus $\trace (J_{\bar z_b^*})=\trace (J_{\bar z_{b^{\dagger}}})=0$ and $\det (J_{\bar z_b^*})>0$ and $\det (J_{\bar z_{b^{\dagger}}})>0$. By Remark~\ref{obs:im-eigs-char} at these two values of $b$ the equilibrium has two non-zero conjugate purely imaginary eigenvalues.  The transversality condition of the Hopf bifurcation follows from the fact that $\bar z_b$ is strictly increasing with respect to $b$, which also implies uniqueness of critical values $b^*$ and $b^{\dagger}$. By Thm.~$3.4.2$ in \cite{Guckenheimer_Nonlinear_2013} we conclude that at  $b=b^*$ the unique equilibrium goes from being a stable focus to an unstable focus in a Hopf bifurcation, and at $b=b^{\dagger}$ it goes from being an unstable focus to a stable focus, also in a Hopf bifurcation.
\end{proof}

Even though from the previous analysis we cannot conclude whether the Hopf bifurcations in Thm.~\ref{thm:hopf} are supercritical or subcritical, we have the following result, which guarantees the existence of a limit cycle precisely in the interval after the first Hopf and before the second Hopf.

\begin{theorem}[Limit Cycle After Input Threshold]\label{thm:hopf-LC}\ \\
    Assume the same hypothesis as in Thm. \ref{thm:hopf}. Then for input values $b\in(b^*, b^{\dagger})$ there exists a limit cycle.
\end{theorem}
\begin{proof}
    First, $b\in(b^*,b^{\dagger})\Rightarrow \bar z_b\in(z^{\ast},z^{\dagger})$ $\Rightarrow \trace (J_{\bar z_b})>0$, which together with $\det (J_{\bar z_b})>0$ implies the unique equilibrium is unstable. Next, we show the dynamics of the system are bounded. Let the bounding box be $\mathcal{B}=[-d^{-1},d^{-1}]\times[-1, k_sd^{-4}] $. We show for one of the edges of the box that the vector field (\ref{eq:sys}) points towards the box interior. If $(z, s)\in (-d^{-1},d^{-1})\times\{k_sd^{-4}\}$ then $\dot s = -k_sd^{-4} + k_s z^4 = k_s(z^4 - \frac{1}{d^4}) < k_s(\frac{1}{d^4} - \frac{1}{d^4}) = 0$, and so the vector field points towards the interior of the box. The same kind of argument applies to the other three edges.
    We now consider a corner point: if $s=-1,z=d^{-1}$ then $\dot s > 0, z < 0$. This follows similarly for the other corners. 
    Thus, all orbits starting on the boundary of $\mathcal{B}$ stay in $\mathcal{B}$. Noting that $\bar z_b\in \text{int}({\mathcal{B}})$ is an unstable source and there are no other equilibria, we can apply the Poincaré-Bendixson theorem \cite{Guckenheimer_Nonlinear_2013}, which tells us that a compact subset of the plane that is forward invariant and contains only a hyperbolic source must also contain a limit cycle, and every trajectory starting in the set converges to a limit cycle.
\end{proof}



The case where $\mu_0$ is less than, but close to $\frac{d}{a}$ such that $\det(J_z)<0$, presents some additional subtleties in its analysis as discussed in Remark~\ref{rem:range_of_mu}. However, the behavior is similar to that described in Thm.~\ref{thm:hopf}. In summary, the equilibrium can no longer be shown to be unique, and instead there appear three equilibria with the same sign as the input: one stable node, one unstable node, and one saddle. The saddle and unstable node appear in a saddle-node (SN) bifurcation, and then the saddle and stable node annihilate in another SN, which marks the threshold to spiking. It is noteworthy that in this case, the spike limit cycle appears with a non-zero amplitude at the bifurcation, and its existence follows from the same arguments in Thm. \ref{thm:hopf-LC}. In both cases, the amplitude of the spiking limit cycles for inputs above threshold can be studied in the limit $\varepsilon\to 0$ using geometrical singular perturbation methods~\cite{Wechselberger_Canards_2007}.

\section{Spiking onset by increasing $\mu_0$ when $b=0$}\label{sec:PF-bif}

\begin{figure}
    \centering
    \includegraphics[width=1\linewidth]{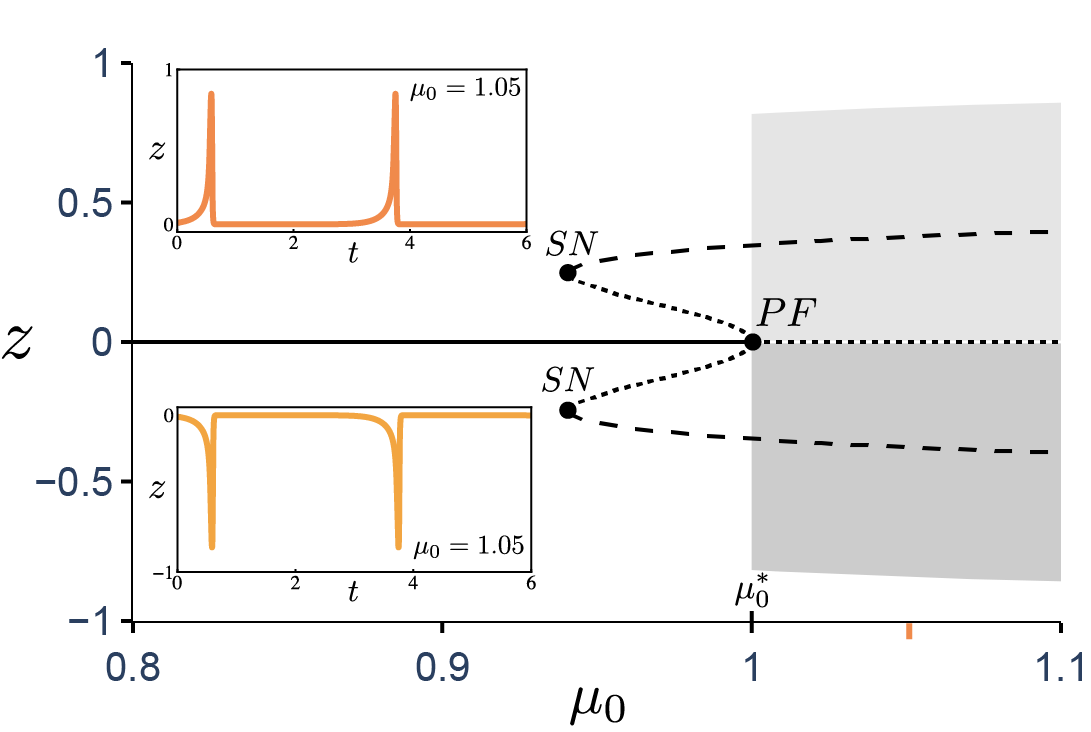}
    \caption{Bifurcation diagram of system (\ref{eq:sys}) with respect to basal sensitivity ($\mu_0$) projected onto the opinion state ($z$) axis to reduce from three dimensions to two. Solid lines indicate stable nodes, dashed lines indicate unstable nodes, dotted lines indicate saddle points, shaded areas indicate the two stable limit cycles in different shades of gray. The vertical width of shaded areas represent amplitude of spike. Black circles represent bifurcations, one subcritical pitchfork ($PF$) and two symmetric saddle-node bifurcations ($SN$). Inset figures show positive (top, $z(0)=0.1$) and negative (bottom, $z(0)=-0.1$) spike limit cycles for $\mu_0=1.05$. Parameters used are $b = 0, a = 1, d = 1, k = 2.3, k_{s} = 16, \varepsilon = 0.1$.}
    \label{fig:munot-bif}
\end{figure}

We now consider the zero input case $b=0$, and study bifurcations of~(\ref{eq:sys}) with respect to $\mu_0$ at the critical point $\mu_0^*=\frac{d}{a}$. This bifurcation is a subcritical pitchfork at the origin, as shown in Fig.~\ref{fig:munot-bif}. This can be shown by projecting the local dynamics around the origin onto the bifurcation center manifold, which reduces to the single-agent non-excitable fast and flexible nonlinear opinion dynamics \cite{Leonard_NOD_ARpaper2024}. For critical and supercritical values of $\mu_0$ there are exactly three equilibria, one of which is the origin. The other two are symmetric with respect to the axis $z=0$, and their properties determine the global behavior of the system. We prove conditions under which each of the two non-zero equilibria is an unstable focus at $\mu_0=\mu_0^*$, which ensures that for some open interval above $\mu_0^*$ the system has two symmetric spike limit cycles, 
one around each unstable focus.
These two large amplitude limit cycles 
appear through a homoclinic bifurcation at the subcritical pitchfork, and not a Hopf as in the previous section. Fig.~\ref{fig:munot-bif} shows a representative bifurcation diagram projected onto the opinion subspace. In an interval above $\mu_0^\ast$, upward and downward spiking oscillations coexist, and even small evidence in one direction or the other can change the sign of the spiking.  

These spikes break deadlock, i.e., a decision is made even with little or no evidence for a superior option. 
In embodied scenarios, self-reinforcement of the randomly chosen option would follow. 
For example, in the navigation problem in Fig. \ref{fig:CharlotteFig1}, if the robot is heading towards an obstacle head-on, and $\mu_0$ grows with the urgency of making a decision, then when $\mu_0$ becomes supercritical the system will produce spikes in one direction in response to even the smallest noise. Once the first spikes happen, the robot will be facing the obstacle at a slight angle, and this will bias the dynamics to keep choosing that direction if there is no further change in information. 

To analyze this bifurcation, first observe that the stability of the origin changes at $\mu_0^*$: when $\mu_0<\frac{d}{a}$ then $c_0=\mu_0 a - d<0$ so $\det(J_0)>0$ and $\trace (J_0) <0$, which implies the equilibrium at the origin is a stable node. When $\mu_0>\frac{d}{a}$ we have $\mu_0a-d>0$ so $\det (J_0)<0$, which implies the equilibrium is a saddle point. In all propositions in this section we assume $b=0$. 

\begin{proposition}
    There are no limit cycles encircling the neutral equilibrium. If the parameters are such that the origin is the only equilibrium, then there are no limit cycles.
\end{proposition}
\begin{proof}
    When $b=0$, the line $z=0$ is invariant (it is the stable manifold of the neutral equilibrium at the origin). If there existed a limit cycle surrounding the origin it would intersect the set $z=0$, but this leads to a contradiction. Suppose the origin is the only equilibrium. Index theory tells us that every limit cycle encircles an equilibrium. Thus, there are no limit cycles. 
\end{proof}

Next we show that for supercritical values of $\mu_0$ the system has exactly three equilibria. 
\begin{proposition}
    If $\mu_0\geq\frac{d}{a}$ and $k>\frac{d^3}{3a}$ then~(\ref{eq:sys}) has exactly three equilibria. 
\end{proposition}
\begin{proof}
Since $h(0)=0$, the Taylor expansion of $h$ at $z=0$ is $h(z)=(a \mu_0-d)z+a(k-\frac{a^2 \mu_0^3}{3})z^3+\mathcal{O}(z^5)$. So $h'(0)=a \mu_0 - d\geq 0$, $h''(0)=0$, $h'''(0)=6a(k-\frac{a^2 \mu_0^3}{3})>0$. This means $h$ is monotonically increasing for some open interval above zero. Observe that $\lim_{z\to\infty}h(z)=-\infty$, so there must be at least one root greater than zero. By a similar argument to Prop. \ref{prop:single-sol}, we show there are no more than three roots of $h$. By (\ref{eq:h'_and_derJ}), $h(\bar z)=0$ implies $h'(\bar z)=-\varepsilon^{-1}\det (J_{\bar z})$. By considering it as a third degree polynomial of $\bar z^2$, it is straightforward to show that, under the hypothesis in the statement, $\det (J_{\bar z})$ changes sign exactly once in the interval $(0, d^{-1})$. Doing algebra, we can show that the unique root of $\det (J_{\bar z})$ in the interval is not a root of $h$. Together, this implies there is at most one root of $h$ in the interval $(0, d^{-1})$. Using symmetry and the fact that there can be no solution greater than $d^{-1}$, we conclude there are exactly three solutions.
\end{proof}

To prove the existence of limit cycles for $\mu_0>\mu_0^*$, we again make use of the Poincaré-Bendixson theorem. At  $\mu_0=\mu_0^\ast$ and at the limit $\varepsilon=0$, we get $c_0=0$ in (\ref{eq:traceJ}), so solving for the roots of $\trace(J_z)$ reduces to solving a second degree polynomial. Let $\varrho_\pm:=\sqrt\frac{c_2\pm\sqrt{c_2^2-4c_1c_3}}{2c_3}$ be the positive roots of $\trace(J_z)$.
We next give a sufficient condition for the existence of limit cycles.

\begin{theorem}[Spike Limit Cycles After Critical $\mu_0$]\ \\
    If $\frac{d^3}{3a}<k<k_s\varrho_-^2$ then for sufficiently small $\varepsilon>0$ there is an open interval above $\mu_0^*$  such that there exist two symmetric limit cycles, each of which encircles one of the two equilibria that are not the origin. 
\end{theorem}
\begin{proof}
    We show existence in the positive halfspace (the negative follows from symmetry). First, we show that the positive equilibrium is unstable at $\mu_0=\mu_0^*$. The value $\sqrt{k/k_s}$ is the unique positive root of $-dz+\psi(z)$. The polynomial $\psi(z)$ also has exactly three roots. Let $\varsigma$ denote its positive root, with $\sqrt{k/k_s}\leq\varsigma\leq d^{-1}$. Since $x>0\Rightarrow \tanh(x)<x$, then $\forall z \in(0,\varsigma), h(z)<-dz+\psi(z)$. So $h(\bar z)=0<-dz+\psi(\bar z)$ implies $\bar z < \sqrt{k/k_s}$. Condition $\frac{d^3}{3a}<k$ implies $\varrho_-\in(0,d^{-1})$, and so $\forall z\in(0,\varrho_-), \exists \varepsilon>0$ such that $\trace (J_z)>0$. Since by hypothesis $\sqrt{k/k_s}<\varrho_-$, we conclude for $\varepsilon>0$ small enough $\trace(J_{\bar z})>0$.  Now, at $\mu_0=\mu_0^*$, $\det (J_z)$ is negative between $0$ and its first positive root, so it must be that $\det (J_{\bar z})>0$. Otherwise, $0=h(\bar z) = \int_{t=0}^{\bar z} h'(t)dt\leq \int_{t=0}^{\bar z} \det (J_t)dt<0$ which is a contradiction. We conclude $\bar z$ is unstable for $\mu_0$ in some open interval above $\mu_0^*$. 
    The construction of a bounding region is slightly technical, but can be justified rigorously. We outline the following steps. First take the segment from $(d^{-1}, -1)\in\R^2$ to $(d^{-1}, k_sd^{-4})$, and from here to a point $(\varpi, k_sd^{-4})$ such that when continuing the orbit from that point it first crosses the line $s=\mu_0-\mu_0^*$ at a point $(\tilde\varpi, \mu_0-\mu_0^*)$ such that $\dot z>0$ and $\tilde\varpi<\bar z$. This is possible because the horizontal $z=0$ is the stable manifold of the origin
    Include the orbit from $(\varpi, k_sd^{-4})$ to $(\tilde\varpi, \mu_0-\mu_0^*)$ as part of the boundary. Then include the line from $(\tilde\varpi, \mu_0-\mu_0^*)$ to $(\tilde\varpi, -1)$. Finally, include the line segment from $(\tilde\varpi, -1)$ to $(d^{-1}, -1)$. This region contains the equilibrium by construction, and the vector field along the boundary points inward. Thus, we can apply the Poincaré-Bendixson theorem to conclude that there exists a limit cycle and that every trajectory converges to a limit cycle.
\end{proof}

The resulting limit cycles appear through global bifurcations of two homoclinic orbits to the origin at $\mu_0=\mu_0^*$. Thus, close to the critical $\mu_0^*$ the limit cycles have arbitrarily small frequency. This is a desirable property. It implies that while in the slightly supercritical $\mu_0$ regime (i.e., $\mu_0>\mu_0^\ast$), even small noise can lead to spiking, this will be at a low frequency, such that it will break symmetry without making the system overly sensitive to noise. 

\begin{figure}
    \centering
    \includegraphics[width=1.0\linewidth]{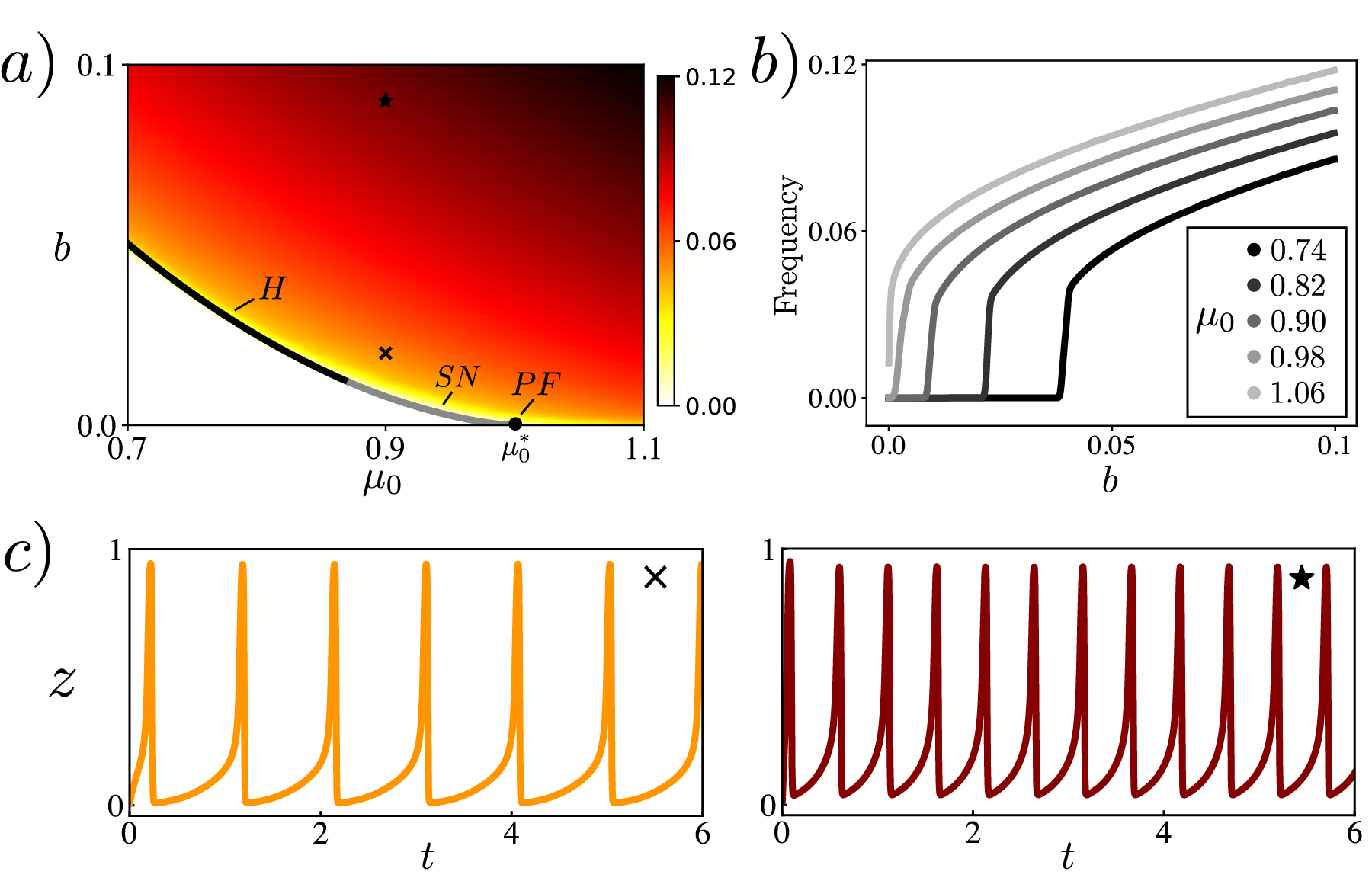}
    \caption{\textbf{a)} Region of the $(\mu_0, b)$ parameter space. For each point, color represents the frequency of spiking. Bifurcations that trigger spiking are: Hopf (black line, $H$), Saddle-Node (gray line, $SN$), and Pitchfork (black point, $PF$). Only positive opinion halfspace is shown to save space. The negative side is the reflection along horizontal. \textbf{b)} Five input magnitude to spiking frequency curves, corresponding to vertical slices of a). \textbf{c)} Spiking for $\mu_0=0.9$ and two values of $b$, marked with cross and star in a). For all subfigures, parameters are $a = 1, d = 1, k = 3, k_{s} = 22, \varepsilon = 0.05$.}
    \label{fig:freqs}
\end{figure}

\section{Tunability of input threshold and frequency encoding}\label{sec:tunability&freq}


In Sections~\ref{sec:Hopf-bif} and \ref{sec:PF-bif} we rigorously showed the identity of the bifurcations that happen at the onset of spiking in S-NOD and studied their thresholds. These happen through three routes. One is through a Hopf bifurcation for supra-threshold inputs at fixed intermediate values of $\mu_0$, which leads to the rapid appearance of a single spike limit cycle in either the positive or negative opinion halfspaces, depending on the sign of the input. The second, described briefly, happens for supra-threshold inputs at close to critical values of $\mu_0<\mu_0^*$ through two Saddle-Node bifurcations in quick succession. The third is through a pitchfork bifurcation that happens at the origin when there is no input and the parameter $\mu_0$ crosses the critical value $\mu_0^*$. Representative bifurcation diagrams of these routes are shown in Figs. \ref{fig:b-bifurcation} and \ref{fig:munot-bif}. 

A fuller picture can be obtained by considering both parameters simultaneously, as shown in Fig. \ref{fig:freqs}a, where the input threshold as well as its dependence on $\mu_0$ can be observed. In the figure, the area in white corresponds to the case of a single (globally) stable equilibrium, whereas the colored region corresponds to a single (globally) attracting spike limit cycle. The figure suggests that the pitchfork bifurcation route is rather special, in the sense that most trajectories along the $(\mu_0, b)$ parameter space going from the resting to spiking regions pass through a Hopf or SN bifurcation. In this sense these routes are the generic ones. Fig.~\ref{fig:freqs}a also shows the dependence of the input threshold $b^*$ on the parameter $\mu_0$. As $\mu_0$ increases $b^*$ becomes smaller. 

\begin{proposition}\label{prop:threshold-shrinks}
    For $\mu_0<\mu_0^*$ such that the input threshold $b^*$ exists and for sufficiently small $\varepsilon>0$, the input threshold is decreasing with respect to $\mu_0$, i.e. $\frac{\partial b^*}{\partial \mu_0}<0$. 
\end{proposition}
\begin{proof}
    From equation \ref{eq:imp-sol} it follows by implicit differentiation that $\frac{\partial b^*}{\partial \mu_0} = \frac{\det(J_{z^{\ast}})}{1-d^2z^{*2}}\frac{\partial z^{\ast}}{\partial \mu_0}$, where $z^{\ast}$ is as in Prop. \ref{prop:trace-roots}. By comparing polynomials $\trace(J_z)$ and $\det(J_z)$ it can be seen that $\forall z\in(-d^{-1}, d^{-1}), -\varepsilon^{-1}\det(J_z) < \trace(J_z)$ for sufficiently small $\varepsilon>0$, thus $ \frac{\det (J_{z^{\ast}})}{1-d^2z^{\ast2}}>0$. Then, $\frac{\partial \trace(J_z)}{\partial \mu_0}=a(1-d^2z^2)$ implies $\frac{\partial z^{\ast}}{\partial \mu_0}<0$. So we conclude $\frac{\partial b^*}{\partial \mu_0}<0$.
\end{proof}

\begin{figure}
    \centering
    \includegraphics[width=0.8\linewidth]{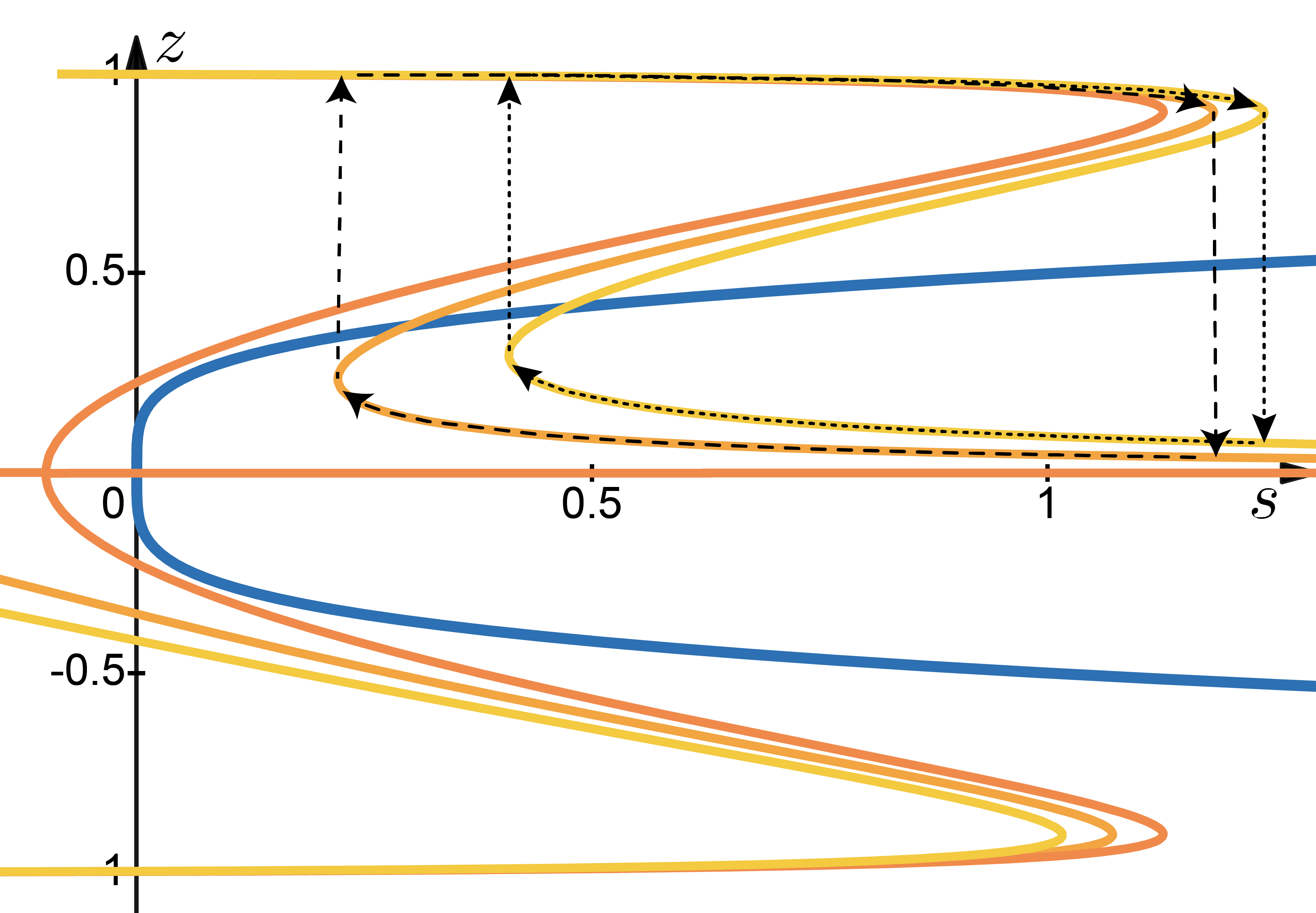}
    \caption{Nullclines for \eqref{eq:sys}. Three nullclines $z$ for values of $b$ equal to $0$ (dark orange), $0.05$ (light orange) and $0.1$ (yellow). Dashed and dotted lines correspond to spike limit cycle for $b=0.05$ and $b=0.1$ respectively, both at the singularly perturbed limit. Blue line is $s$ nullcline. 
    }
    \label{fig:nulls}
\end{figure}
 
The previous proposition covers the case when the equilibrium is unique and a Hopf marks the threshold to spiking (black line in Fig. \ref{fig:freqs}). When the equilibrium is not unique and there is a SN bifurcation (gray line in Fig. \ref{fig:freqs}), similar arguments to Prop. \ref{prop:sol-lim} and Prop. \ref{prop:threshold-shrinks} would apply, but the critical input would be when the smallest equilibrium crosses the first root of $\det(J_z)$. 

S-NOD also displays encoding of the input magnitude as the spiking frequency, which can be observed both in Fig. \ref{fig:freqs}a as the color intensity in the spiking region, and also in Fig. \ref{fig:freqs}b where ``slices" are taken for five fixed values of $\mu_0$. The frequency is always an increasing function of the input for the parameter values studied and tuned by $\mu_0$. Consistent with this, $\det(J_z)$ can be shown to be monotonically increasing in the interval between the onset of spiking and the Hopf at which the system saturates and the limit cycle disappears ($b^\dagger$ in Section \ref{sec:Hopf-bif}, which always exists). Recalling that the instantaneous frequency at a planar Hopf is the square root of $\det(J)$, we have that the initial frequency at the first Hopf (when the fixed point is unique) is smaller than the frequency at the second. Proving monotonicity across the whole interval analytically is not trivial, but geometric singular perturbation arguments could be used to estimate the traveling time of the slow portion of spiking trajectories, which are along the branches of the $z$-nullcline (see, e.g.,~\cite{sacre_singularly_2016}). Fig. \ref{fig:nulls} shows an example of how the $z$-nullcline changes with increasing input. The spike limit cycles shown are at the singular perturbation limit ($\varepsilon\to0$), for two suprathreshold values of $b$.

The other parameters provide different dials which can also be used to tune the thresholds and other properties of the spiking behavior, such as the amplitude and the duration of individual spikes. In particular, parameter $d$ can tune the amplitude of the spikes, while parameters $a$, $k$ and $k_s$ can tune the duration and frequency encoding. Even though changes to these parameters are not independent, appropriate simultaneous changes in these can decouple the spike properties and allow for freedom in the design for particular applications. Fig. \ref{fig:b-bifurcation} shows a set of parameters where $b^*$ and $z^{\ast}$ are comparable in magnitude to the spike amplitude; however, these can be designed to be of arbitrary size, such that they are separated by orders of magnitude. This could be important in applications with noisy measurements, where it is important to clearly distinguish sub-threshold states versus spiking states. 



\section{Discussion}
The results in this paper show the tractability and robustness of the S-NOD model properties such as input thresholds to spiking, indecision-breaking and frequency encoding. These properties can be leveraged in order to use S-NOD as a new neuromorphic control block for adaptable robustness and agility in a way that generalizes the approach of \cite{Cathcart_SNOD_2024}. Filtered input, e.g., measurement of the plant output and environment, enters the block as a real-valued signal $b(t)$, and the opinion state $z(t)$ is filtered and transformed into a control signal for the plant that is fed to the actuators. The parameter $\mu_0$ can be adjusted as a function of some aspect of the environment in order to modulate the input threshold, for example in response to urgency or a need for higher responsiveness. S-NOD can be extended to consider multiple agents and/or multiple options similarly to how the NOD model generalizes in \cite{Bizyaeva_NOD_2023}. This will allow for rich collective dynamics and for multidimensional opinion spikes enabling more complex actuation.

\bibliographystyle{IEEEtran}
\bibliography{refs}

\end{document}